\documentclass[journal]{IEEEtran}
\usepackage{amssymb}
\usepackage[cmex10]{amsmath}
\usepackage{stfloats}
\usepackage{graphicx}
\usepackage{subfigure}
\usepackage{tabularx}
\usepackage{epsfig,epsf,color,balance,cite}
\usepackage{verbatim}
\usepackage{url}
\usepackage{bm}
\usepackage{amsthm}
\newtheorem{lemma}{Lemma}
\usepackage{algorithm}
\usepackage{algorithmic}

\usepackage{color}
\definecolor{myc1}{rgb}{0,0,0}

\begin{document}
	\title{{A Joint Communication and Computation Framework for Digital Twin over Wireless Networks} %in Metaverse
	}
	\author{
		Zhaohui Yang,
			Mingzhe Chen,
			Yuchen  Liu, and
			Zhaoyang~Zhang
		\thanks{Z. Yang and Z. Zhang are with the College of Information Science and Electronic Engineering, Zhejiang University, Hangzhou, Zhejiang 310027, China, and Zhejiang Provincial Key Lab of Information Processing, Communication and Networking (IPCAN), Hangzhou, Zhejiang, 310007, China. Z. Yang is also with Zhejiang Lab,  Hangzhou, Zhejiang, 311100, China. (e-mails: yang\_zhaohui@zju.edu.cn, ning\_ming@zju.edu.cn) }
  \thanks{M. Chen is with the Department of Electrical and Computer Engineering and Institute for Data Science and Computing, University of Miami, Coral Gables, FL, 33146 USA.  (e-mail:mingzhe.chen@miami.edu)}
		\thanks{Y. Liu is with Department of Computer Science, North Carolina State University, Raleigh, NC, 27606, USA. (e-mail:yuchen.liu@ncsu.edu)}
	}

	\maketitle
%	\vspace{-0.25em}
	\begin{abstract}
In this paper, the problem of low-latency communication and computation resource allocation for digital twin (DT) over wireless networks is investigated. In the considered model, multiple physical devices in the physical network (PN) needs to frequently offload the computation task related  data to the digital network twin (DNT), which is generated and  controlled by the central server. Due to limited energy budget of the physical devices, both computation accuracy and wireless transmission power must be considered during the DT procedure. This joint communication and computation problem is formulated as an optimization problem  whose goal is to minimize the overall transmission delay of the system under total PN energy  and DNT model accuracy constraints. To solve this problem, an alternating algorithm with iteratively solving device scheduling, power control, and data offloading subproblems. For the device scheduling subproblem, the optimal solution is obtained in closed form through the dual method. For the special case with one physical device,  the optimal number of transmission times is reveled. Based on the theoretical findings, the original problem is transformed into a simplified problem and the optimal device scheduling can be found.  Numerical results verify that the proposed algorithm can reduce the transmission delay of the system by up to 51.2\% compared to the conventional schemes. 
		
		%a physical network with multiple devices and its mapping digital network twin are involved. 
	\end{abstract}
	\IEEEpeerreviewmaketitle
	
	\section{Introduction}
Metaverse, considered as a new generation of the Internet, is envisioned to build a digital world where people can meet and interact in real time via integrating various technologies, such as extended reality (XR), digital twin (DT), holographic, sensing, communication, and computing \cite{khan2022digital,jiang2022towards,tang2022roadmap,cai2022compute,wang2022survey,xu2022full,chang20226g,sun2016edgeiot}. 	
Recently, {\color{myc1}DT technology is envisioned to act as an  important role for the modern communication society \cite{saad2019vision,jones2020characterising,liu2021review,xu2022edge,zhang2018optimal}, in particular for the future applications including Metaverse, as shown in Fig.~\ref{fig0}.} DT is the process of using information technology to digitally define and model physical entities. The core concept is to realize feedback optimization of physical entities through the simulation, control and prediction of DTs. Due to the combination of digital and physical worlds,  DT technology  has many advantages \cite{he2018preserving,han2013optimizing,zhou2020pirate,qu2021service,arnold2020cooperative,le2022survey,arnold2019survey,zhan2020big}. 
The core element of a DT is data. It originates from physical entities, operating systems, sensors, etc. It covers simulation models, environmental data, physical object design data, maintenance data, operation data, etc., and runs through the entire operation of physical objects.
The DT is used as a data storage platform to collect various raw data, perform data fusion processing, promote the dynamic operation of each part of the simulation model, and effectively reflect various business processes. Therefore, data is the fundamental of DT applications, and without multivariate fusion data, DT applications will lose momentum.
Moreover, the main body of a DT system is a data-driven model built on the logic of physical entities and behavior. Twin data is the basis for mapping between physical objects and digital world model objects. It includes models, behavioral logic, business processes, state changes, etc., to achieve comprehensive presentation, accurate expression and dynamic monitoring of the state and behavior of physical entities in the digital world.
DT can be used to intelligently design decisions, using the large amount of historical, real-time data in the DT, combined with advanced algorithm models, to effectively reflect the state and behavior of physical objects in the digital world.
Thus, the data scurity is important for DT \cite{yang2023secure}.
At the same time, through the simulation experiment and analysis and prediction in the digital world, it provides a decision-making basis for the instruction formulation of entity objects and the further optimization of the process system, which greatly improves the efficiency of analysis and decision-making
\cite{yu20226g,akyildiz2022metaverse,chen2023big}.
%DT technology requires the construction of digital representations of physical objects in digital space. Physical objects in the real world and twins in digital space can achieve bidirectional mapping, data connection and state interaction. 
%Moreover, based on the acquisition of multiple data such as real-time sensing, the twin body can comprehensively, accurately and dynamically reflect the state changes of physical objects, including appearance, performance, location, anomaly, etc. Besides, in an ideal state, the mapping and synchronization state achieved by the DT should cover the entire life cycle of the twin object from design, production, operation to retirement, and the twin body should continue to evolve and update with the life cycle process of the twin object.
%The ultimate goal of establishing a twin is to describe the internal mechanism of the physical entity, analyze laws, gain insight into trends, and form optimization instructions or strategies for the physical world based on analysis and simulation, so as to achieve a closed-loop decision-making optimization function for physical entities \cite{yu20226g,akyildiz2022metaverse}.
 
	\begin{figure*}
		\centering
		\includegraphics[width=6.5in]{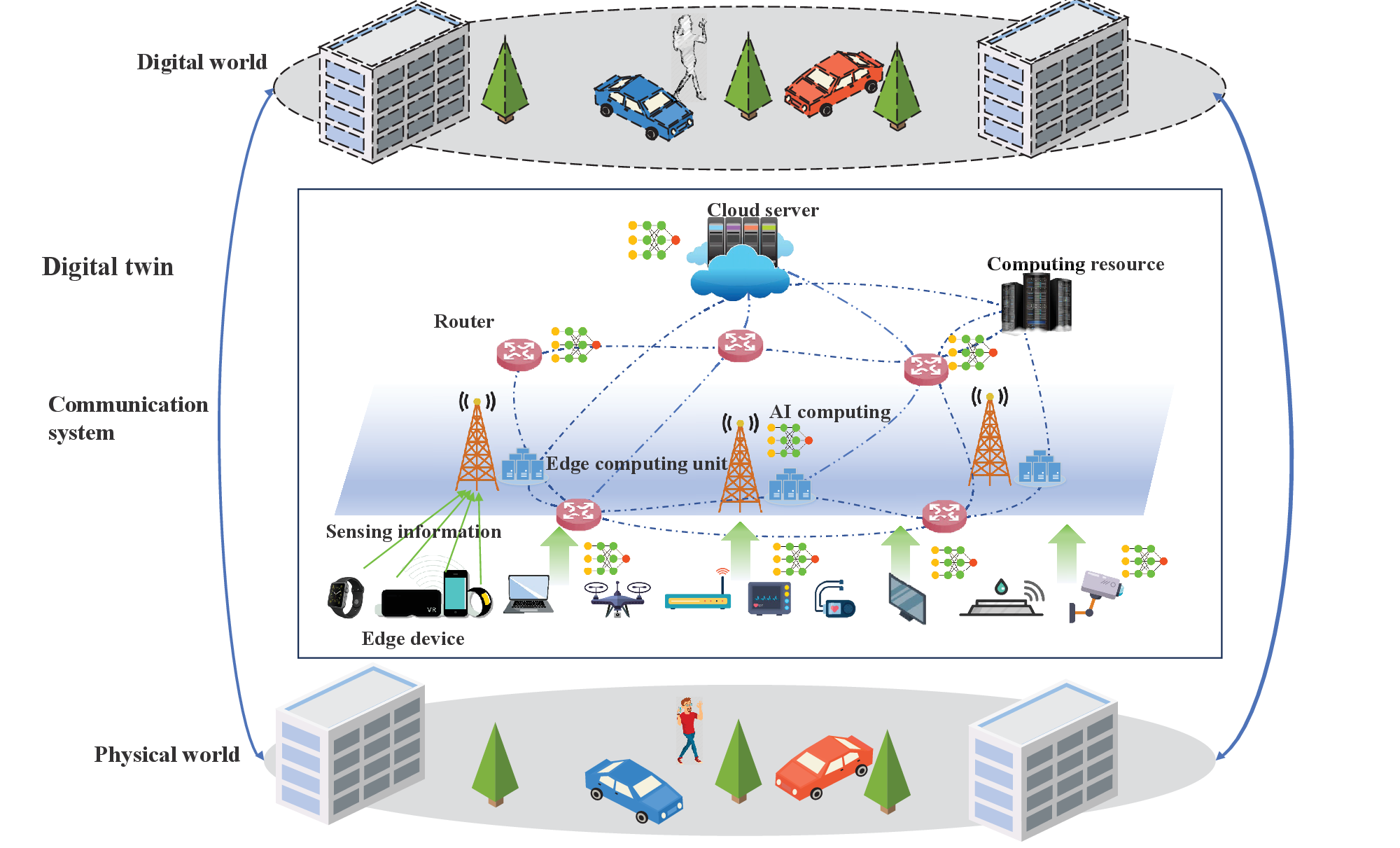}
		%\vspace{-2em}
		\caption{ {\color{myc1} The application scenario of DT based communication system.}}\label{fig0}
		%\vspace{-2em}
	\end{figure*}	
Due to the above distinctive advantages, DT has many emerging applications \cite{jones2020characterising,zhan2020learning,wang2022task,wang2023task}. Implementing DT requires a joint design of the physical and application layer, which can involve the multi-tier framework \cite{wang2022task,wang2023task}. 
With the help of high-tech means such as geographic information technology and three-dimensional virtualization, the smart DT community restores the situation of community buildings and traffic roads with high accuracy, integrates two-dimensional models of each floor of the building, indoor and outdoor integration, and displays the key data of the core operation system of the community in real time from the scene display of the three dimensions of smart operation and maintenance, intelligent transportation, and digital life.
 Based on the visual data, a city-level DT system can be built on the basis of fully integrating the information resources of various fields of the city, accurately reproducing the management elements of a wide range of urban fields, and realizing all-round dynamic perception of a wide range of urban operation from the global perspective to the micro field.

%Virtual ``clones'' of these physical operations can help organizations monitor operations, perform predictive maintenance and provide insight for capital purchase decisions. They can also help organizations simulate scenarios that are too time-consuming or expensive to test with physical assets, create long-term business plans, identify new findings and improve processes. DTs can help companies virtually test and validate products before they come out. Engineers can use them to identify process faults. Organizations can use DTs to proactively monitor equipment and systems to schedule maintenance before they fail, increasing productivity. Users can monitor and control systems remotely. Process automation and 24x7 access to system information allow technicians to have more time to focus on collaboration. By integrating financial data, organizations can use DTs to make better, faster-adjusted decisions.
	
	{\color{myc1}Since the physical world and the digital world needs to communicate, DT over wireless networks has attracted a lot of attentions \cite{khan2022digital}. 
	The survey of using the future sixth generation (6G) communication techniques for DT was presented in \cite{kuruvatti2022empowering}.
	A blockchain empowered federated learning framework for DT was considered in \cite{lu2020low,lu2021adaptive} to solve the edge association problem through optimizing the DT association,
	training data batch size, and bandwidth allocation. 
	The interplay between Terahertz and DT was considered in \cite{pengnoo2020digital}, where DT can be utilized to predict and simulate the  unique propagation properties of Terahertz signals.
	Furthermore, a DT  assisted mobile edge computing network was considered in 
	\cite{zheng2021learning} for Internet of vehicles. 
	In \cite{ruah2022digital}, the combination of Bayesian learning of DT was studied, where the DT trained a  Bayesian model to predict the epistemic uncertainty of the wireless communication system. 
	In \cite{hashash2022edge}, a DT system over wireless communication network was proposed to investigate the tradeoff between the accuracy and  delay of the DT system. 
	An weighted sum accuracy and system delay optimization problem was formulated in  \cite{hashash2022edge}, which was solved by using the edge continual learning. 
	However, the above works \cite{lu2020low,lu2021adaptive,pengnoo2020digital,zheng2021learning,ruah2022digital,kuruvatti2022empowering,hashash2022edge} all ignored the joint communication and computation resource allocation with  considering energy budget of wireless devices and model accuracy of the DT, even though the wireless devices are usually energy constrained. }
	
	In this paper, we consider the delay-efficient communication and computation resource allocation for a DT network with considering energy budget of wireless devices and DT model accuracy constraints. 
	Our contributions are listed as follows.
	\begin{itemize}
		%\vspace{-0.35em}
		\item
		We investigate the performance of DT over wireless networks, where multiple physical devices in the physical network (PN) transmit data to the digital network twin (DNT) over multiple time slots. Due to limited energy budget of the physical devices, both computation accuracy and wireless transmission power are  considered during the DT procedure.
		\item   This joint communication and computation problem is formulated as an optimization problem  whose goal is to minimize the overall transmission delay of the system under total PN energy  and DNT model accuracy constraints. To solve this problem, an alternating algorithm with iteratively solving device scheduling, power control, and task offloading subproblems. 
  \item For the special case with one physical device, we reveal the optimal number of transmission times. The optimal transmit power is derived in closed form and the original problem is then transformed into a simplified problem. 
		\item Numerical results show the superiority of the proposed algorithm compared to the conventional schemes in terms of transmission delay. 
	\end{itemize}

	The rest of this paper is organized as below.
	Section \uppercase\expandafter{\romannumeral2} presents the system model and problem formulation.
	The algorithm design is presented in Section \uppercase\expandafter{\romannumeral3}, while simulation results are given in Section \uppercase\expandafter{\romannumeral4}.
	Section \uppercase\expandafter{\romannumeral5} concludes this paper.

	\section{System Model and Problem Formulation}
	\subsection{System Model}
	\begin{figure}
		\centering
		\includegraphics[width=3.3in]{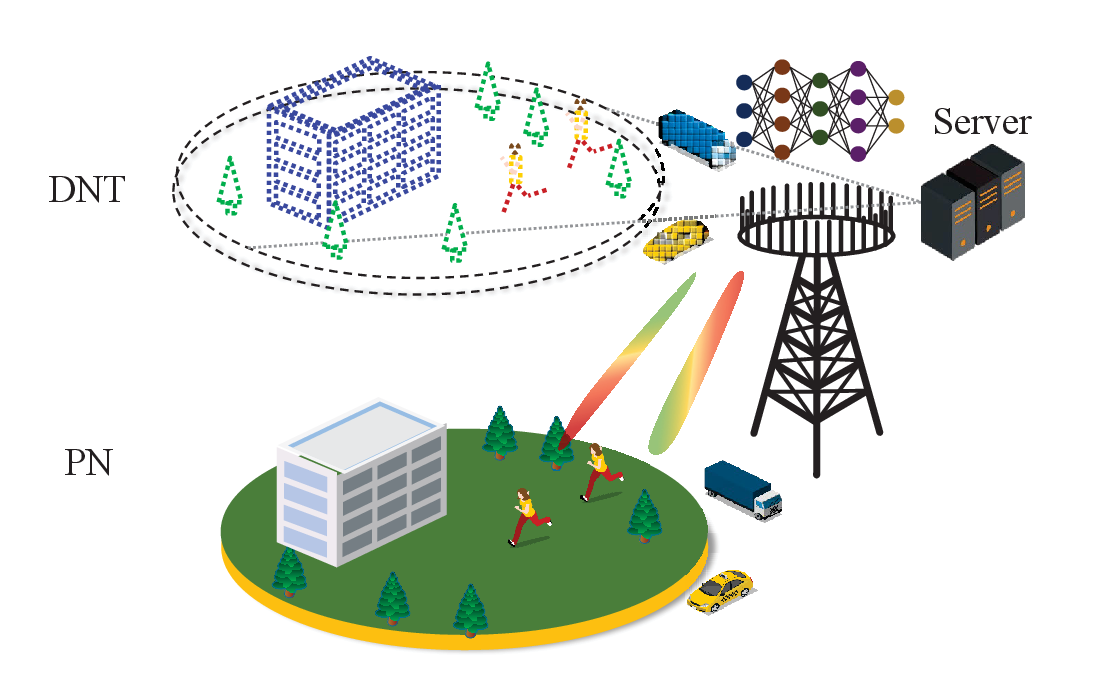}
		%\vspace{-2em}
		\caption{The considered system model of DNTs.}\label{sys0}
		%\vspace{-2em}
	\end{figure}
	
	Consider the DNTs that consist of a PN and its mapping DNT generated and controlled by a central server, as shown in Fig. \ref{sys0}.
	The total number of physical devices  in the PN is denoted by $K$. 
	The physical devices  such as base stations  and sensors need to transmit status data to the central server, which utilizes the obtained data to generate DNT.
	We consider a long time period $T$ and the time can be divided into $N$ time slots.
	The duration of each time can be calculated as $T_0=\frac T N$.
	
	The channel gain between physical device $k$  and the server at time slot $n$ is denoted by $h_{nk}$.
	At time slot $n$, the physical device $k$ in PN will generate data $D_{nk}$.
	To guarantee synchronization between the PN and the DNT, 
	the  physical device $k$  can choose to transmit data $d_{nk}$ to the server.
	The transmission rate between the  physical device $k$  and the server is  given by 
	\begin{equation}\label{sys_eq1}
		r_{nk}=B\log_2\left(1+\frac{p_{nk}h_{nk}}{\sigma^2}\right),
	\end{equation}
	where $B$ is the bandwidth of the system, $p_{nk}$ is the transmit power of the physical device $k$ at time slot $n$, and $\sigma^2$ is the noise power.
	Considering the randomness of the wireless channel, the received data at the server can contain error. 
	The error rate of the transmission at time slot $n$ can be given by \cite{chen2021a}
	\begin{equation}\label{sys_eq1_2}
		e_{nk}(p_{nk})=1-e^{-\frac{m\sigma^2}{p_{nk}h_{nk}}},
	\end{equation}
	with $m$ being a waterfall threshold.

	Since the transmitted data cannot be large than that of the remaining data, we have
	\begin{equation}\label{sys_eq2}
		\sum_{i=1}^n d_{ik} c_{ik}(p_{ik})\leq  \sum_{i=1}^n D_{ik}, \quad \forall n\in \mathcal N,k\in\mathcal K,
	\end{equation}
	where  
	\begin{equation}\label{sys_eq2_2}
		c_{ik}(p_{ik})= \left\{ \begin{array}{ll}
			1 & \textrm{with probability $e^{-\frac{m\sigma^2}{p_{ik}h_{ik}}}$}\\
			0 & \textrm{with probability $1-e^{-\frac{m\sigma^2}{p_{ik}h_{ik}}}$}
		\end{array}\right. ,
	\end{equation}
	$\mathcal N=\{1,\cdots, N\}$, and  $\mathcal K=\{1,\cdots, K\}$.
	
	In each time slot,
	let $x_{nk}\in\{0,1\}$ denote whether physical device $k$ in the PN transmits data to the server.
	The notation 
	$x_{nk}=1$ implies that physical device $k$ transmits data to the BS; otherwise we have $x_{nk}=0$.
	In this paper, frequency division multiple access (FDMA) is considered for the uplink transmission. Due to limited resource blocks of the communication system, the maximum number of associated users at each time slot is limited, i.e., 
	\begin{equation}\label{sys_eq2}
		\sum_{k=1}^K x_{nk}\leq K_0, \quad \forall n\in\mathcal N,
	\end{equation}
	where $K_0$ is the maximum number of available resource blocks for the communication system. 
	
	At each time slot $n$, the wireless transmission delay and energy can be derived as
	\begin{equation}\label{sys_eq3}
		t_{nk}=\frac{d_{nk}}{r_{nk}}.
	\end{equation}
	Since multiple users can simultaneously communicate with the server, the transmission delay of the PN at time slot $n$ can be given by
	\begin{equation}
		t_{n}=\max_{k\in\mathcal K}x_{nk}t_{nk}.
	\end{equation}
	
	Moreover, at each time slot $n$, the wireless transmission  energy can be given by 
	\begin{equation}
		e_{nk}=t_{nk} p_{nk}.
	\end{equation}
	
	The accuracy of the DNT requires more data from the PN, while large data can lead to high transmission delay and energy. 
	Thus, it is of importance to investigate the tradeoff between the accuracy and the wireless cost of the PN including delay and energy. 
	The accuracy of each time slot $n$ can be modeled as
	\begin{equation}\label{sys1accu}
		a_n=f\left(\sum_{i=1}^n \sum_{k=1}^K d_{ik}c_{ik}(p_{ik}),\sum_{i=1}^n \sum_{k=1}^KD_{ik}\right),
	\end{equation}
	where function $f(x,y)\in[0,1]$ and $f(x,y)$ increases with $x$ while decreases with $y$.
	The function of $f(x,y)$ can be obtained through simulations such as \cite{hashash2022edge}. 
	For example, we can set $f(x,y)=\left(\frac{x}{y}\right)^{\alpha}$, where $\alpha>0$.

	\subsection{Problem Formulation}
	With the considered system model, our aim is minimize the transmission delay of the system with both accuracy and energy constraints.
	Mathematically, the optimization problem can be formulated as:
	\begin{subequations}\label{sys1min1}
		\begin{align}
			\mathop{\min }_{\boldsymbol x, \boldsymbol d, \boldsymbol p}\quad
			&\sum_{n=1}^N  \max_{k\in\mathcal K}x_{nk}t_{nk} \tag{\theequation}\\
			\textrm{s.t.}\quad\:\:
			& f\left(\sum_{i=1}^n \sum_{k=1}^K d_{ik}c_{ik}(p_{ik}),\sum_{i=1}^n \sum_{k=1}^KD_{ik}\right) \geq A_n, 
   \\  &\qquad\qquad\qquad\qquad\qquad\qquad\qquad \forall n\in\mathcal N,\\
			& \sum_{n=1}^N t_{nk} p_{nk} \leq Q_k, \quad \forall k\in\mathcal K, \\
			&\sum_{i=1}^n d_{ik} c_{ik}(p_{ik})\leq  \sum_{i=1}^n D_{ik}, \quad \forall n\in \mathcal N,k\in\mathcal K,\\
			&\sum_{i=n}^{n+\tau } x_{ik} \geq  \beta_k \tau, \quad \forall n\in\mathcal N,k\in\mathcal K,  \\
			&	\sum_{k=1}^K x_{nk}\leq K_0, \quad \forall n\in\mathcal N,\\
			&d_{nk}\geq 0, \quad \forall n\in\mathcal N, k\in\mathcal K, \\
   			& x_{nk}\in\{0,1\}, \quad \forall n\in\mathcal N, k\in\mathcal K, \\
			&0\leq p_{nk}\leq P_k,   \quad \forall n\in\mathcal N,k\in\mathcal K,\\
   			& 0\leq t_{nk}\leq T_0, \quad \forall n\in\mathcal N,k\in\mathcal K,
      \end{align}
	\end{subequations}
where $\boldsymbol x=[x_{11},\cdots,x_{1K},\cdots,x_{NK}]^T$,
	$\boldsymbol d=[d_{11},\cdots, $ $d_{1K},\cdots,d_{NK}]^T$,
	$\boldsymbol p=[p_{11},\cdots,p_{1K},\cdots,$ $p_{NK}]^T$, $Q_k$ is the maximum energy of physical device $k$,  $\beta_k\in(0,1]$ is a parameter to ensure that the physical device $k$ and the server should have regular communication, $\tau$ is a constant parameter to ensure that the each physical device and the server must have at least one communication in $\tau$ time slots, and $P_k$ is the maximum transmit power of physical device $k$.
	
	Problem \eqref{sys1min1} is mixed integer optimization problem, which is generally hard to solve due to the following three difficulties. 
	The first difficulty is the complicated accuracy function \eqref{sys1accu}, of which the explicit expression is hard to obtain.
	The second difficulty lies in the nonconvex objective function \eqref{sys1min1} and constraints (\ref{sys1min1}a)-(\ref{sys1min1}c). 
	The third difficulty lies in the integer scheduling variable $x_{nk}$.
	To solve problem \eqref{sys1min1}, we propose an iterative algorithm in the following section.
	\section{Algorithm Design}
	In this section, we present the proposed alternating algorithm to solve problem \eqref{sys1min1}, which alternative solves three subproblems at each iteration., i.e., device  scheduling subproblem, power control subproblem, and data offloading subproblem. 
	\subsection{Device Scheduling Subproblem}
	With given transmission power and data offloading variables in problem \eqref{sys1min1},  the device scheduling subproblem can be given by
	\begin{subequations}\label{alg2min1}
		\begin{align}
			\mathop{\min }_{\boldsymbol x}\quad
			&\sum_{n=1}^N  \max_{k\in\mathcal K}x_{nk}t_{nk} \tag{\theequation}\\
			\textrm{s.t.}\quad\:\:
			&\sum_{i=n}^{n+\tau } x_{ik} \geq  \beta_k \tau, \quad \forall n\in\mathcal N,k\in\mathcal K,  \\
			&	\sum_{k=1}^K x_{nk}\leq K_0, \quad \forall n\in\mathcal N,\\
			& x_{nk}\in\{0,1\}, \quad \forall n\in\mathcal N.    
		\end{align}
	\end{subequations}
	Problem \eqref{alg2min1} is a linear integer optimization problem.
	Relaxing the integer constraints and introducing   slack variables $y$, 
	Problem \eqref{alg2min1} becomes,
	\begin{subequations}\label{alg2min1_2}
		\begin{align}
			\mathop{\min }_{\boldsymbol x,\boldsymbol y}\quad
			&\sum_{n=1}^N  y_n \tag{\theequation}\\
			\textrm{s.t.}\quad\:\:
			&\sum_{i=n}^{n+\tau } x_{ik} \geq  \beta_k \tau, \quad \forall n\in\mathcal N,k\in\mathcal K,  \\
			&	\sum_{k=1}^K x_{nk}\leq K_0, \quad \forall n\in\mathcal N,\\
			&y_n\geq  x_{nk}t_{nk}, \quad \forall n\in\mathcal N,k\in\mathcal K,\\
			& x_{nk}\in\{0,1\} \quad \forall n\in\mathcal N,   
		\end{align}
	\end{subequations}
	where $\boldsymbol y=[y_1,\cdots,y_N]^T$.
	Problem \eqref{alg2min1_2} is a linear integer optimization problem, which is generally hard to obtain the globally optimal solution. 
 In the following, we use the dual method to solve the relaxed problem of \eqref{alg2min1_2} with replacing constraint  (\ref{alg2min1_2}d) with continuous constraint $x_{nk}\in[0,1]$.
 It can be also proved that the obtained solution of the relaxed problem automatically satisfies the integer constraint in  (\ref{alg2min1_2}d).
 %which can be effective solved via the simplex method. Having obtain the continuous value of $x_{nk}$, we use the rounding method to obtain the integer value of $x_{nk}$.
Through replacing replacing constraint  (\ref{alg2min1_2}d) with continuous constraint $x_{nk}\in[0,1]$, problem (\ref{alg2min1_2}) can be transformed into 
	\begin{subequations}\label{alg2min1_2_2}
		\begin{align}
			\mathop{\min }_{\boldsymbol x,\boldsymbol y}\quad
			& \sum_{n=1}^N  y_n  \tag{\theequation}\\
			\textrm{s.t.}\quad\:\:
			&\sum_{i=n}^{n+\tau } x_{ik} \geq  \beta_k \tau, \quad \forall n\in\mathcal N,k\in\mathcal K,  \\
			&	\sum_{k=1}^K x_{nk}\leq K_0, \quad \forall n\in\mathcal N,\\
			&y_n\geq  x_{nk}t_{nk}, \quad \forall n\in\mathcal N,k\in\mathcal K,\\
			& x_{nk}\in[0,1] \quad \forall n\in\mathcal N.  
		\end{align}
	\end{subequations}
%where maximizing the objective function  (\ref{alg2min1_2}) is equivalent to maximize $\left(\sum_{n=1}^N  y_n \right)^2$. 
Since both objective function and constraints of problem \eqref{alg2min1_2_2} is convex, problem \eqref{alg2min1_2_2} is a convex problem and the optimal solution can be obtained by using the dual method. 
To obtain the optimal solution of problem \eqref{alg2min1_2_2}, we provide the following lemma. 
\begin{lemma}
The optimal solution of problem \eqref{alg2min1_2_2} is 
\begin{eqnarray}\label{alg2min1_2Th1eq1}
x_{nk}^* = \left\{ \begin{array}{ll}
	1 & \!\!\!\textrm{if\,\,\,} \lambda_{2n}+\lambda_{3nk}-\sum_{i=\max\{n-\tau,1\}}^n \lambda_{1ik}<0\\
	0 & \!\!\!\textrm{else }
\end{array} \right. \!\!\!,
\end{eqnarray}
and 
\begin{equation}\label{alg2min1_2Th1eq2}
    y_n^*=\max_{k\in\mathcal K}x_{nk}^*t_{nk}, 
\end{equation}
where $\lambda_{1nk}\geq 0$, $\lambda_{2n}\geq 0$, and $0\leq \lambda_{3nk}\leq 1$ are the corresponding dual variables associated with constraints (\ref{alg2min1_2_2}a)-(\ref{alg2min1_2_2}c).
\end{lemma}
\begin{proof}
The dual function of problem \eqref{alg2min1_2_2} can be formulated by 
\begin{eqnarray}\label{alg2min1_2eq1}
	\begin{aligned}
&\mathcal L_1(\boldsymbol x, \boldsymbol y, \boldsymbol \lambda_1, \boldsymbol \lambda_2,\boldsymbol \lambda_3)
\nonumber\\=&  \sum_{n=1}^N  y_n  + \sum_{n=1}^N
\sum_{k=1}^K \lambda_{1nk} \left(\beta_k \tau-\sum_{i=n}^{n+\tau } x_{ik}\right)
\nonumber\\&
+\sum_{k=1}^K \lambda_{2n} \left(\sum_{k=1}^K x_{nk}-K_0\right)
\nonumber\\&+\sum_{n=1}^N
\sum_{k=1}^K \lambda_{3nk} \left(  x_{nk}t_{nk}-y_n\right),
	\end{aligned}
\end{eqnarray}
where $\boldsymbol \lambda_1=[\lambda_{111},\cdots,\lambda_{1NK}]^T$,
$\boldsymbol \lambda_2=[\lambda_{21},\cdots,\lambda_{2N}]^T$, and
$\boldsymbol \lambda_3=[\lambda_{311},\cdots,\lambda_{3NK}]^T$.
Considering constraint (\ref{alg2min1_2_2}d) and $\mathcal L_1(\boldsymbol x, \boldsymbol y, \boldsymbol \lambda_1, \boldsymbol \lambda_2,\boldsymbol \lambda_3)$ is a linear function with respect to $\boldsymbol x$, the optimal solution of $\boldsymbol x$ can be derived as \eqref{alg2min1_2Th1eq1}. 
Further considering that $y_n$ only has lower bound as shown in (\ref{alg2min1_2_2}c),  the dual function $\mathcal L_1(\boldsymbol x, \boldsymbol y, \boldsymbol \lambda_1, \boldsymbol \lambda_2,\boldsymbol \lambda_3)$ is unbounded unless $1-\lambda_{3nk}\geq 0$.
Thus, we always have $1-\lambda_{3nk}\geq 0$ and the optimal solution of $y_n^*$ should be the minimum value satisfying (\ref{alg2min1_2_2}c), as shown in \eqref{alg2min1_2Th1eq2}.
\end{proof}

According to Lemma 1, the optimal solution of $\boldsymbol x$ and $\boldsymbol y$ can be obtained with given dual variables. Besides, the value of $x_{nk}$ is always zero or one, which automatically satisfies the integer constraint (\ref{alg2min1_2}d).
With obtained $\boldsymbol x$ and $\boldsymbol y$, the values of the dual variables can be updated with the sub-gradient method \cite{boyd2004convex}:
\begin{equation}
  \lambda_{1nk}(t+1)=  \left[\lambda_{1nk} (t)+\kappa(t)\left(\beta_k \tau-\sum_{i=n}^{n+\tau } x_{ik}\right)\right]^+,
\end{equation} 
\begin{equation}
 \lambda_{2n}(t+1)=  \left[\lambda_{2n} (t)+\kappa(t)   \left(\sum_{k=1}^K x_{nk}-K_0\right)\right]^+,
\end{equation} 
\begin{equation}
 \lambda_{3nk}(t+1)= \left. (\lambda_{3nk} (t)+\kappa(t) \left(  x_{nk}t_{nk}-y_n\right))\right|_0^1,
\end{equation} 
where $(t)$ denotes the value of variable in the $t$-th iteration, $\kappa(t)$ is the dynamatic stepsize of the dual method, $[x]^+=\max\{x,0\}$, and $x_a^b=\max\{\min\{x,a\},b\}$. 
Through iteratively updating the  value of $\boldsymbol x$ and $\boldsymbol y$ with the dual variables, the optimal solution is obtained. 

	\subsection{Power Control Subproblem}
	With given device scheduling and data offloading variables in problem \eqref{sys1min1},  the  power control subproblem can be formulated as 
	\begin{subequations}\label{alg2min2}
		\begin{align}
			\mathop{\min }_{\boldsymbol p}\quad
			&\sum_{n=1}^N  \max_{k\in\mathcal K}x_{nk}t_{nk} \tag{\theequation}\\
			\textrm{s.t.}\quad\: 
			& a_n \geq A_n, \quad \forall n\in\mathcal N,\\
			& \sum_{n=1}^N t_{nk} p_{nk} \leq Q_k, \quad \forall k\in\mathcal K, \\
			&\sum_{i=1}^n d_{ik} c_{ik}(p_{ik})\leq  \sum_{i=1}^n D_{ik}, \quad \forall n\in \mathcal N,k\in\mathcal K,\\
			&0\leq p_{nk}\leq P_k,   \quad \forall n\in\mathcal N,\\
   			& 0\leq t_{nk}\leq T_0, \quad \forall n\in\mathcal N.
      \end{align}
	\end{subequations}
 
 Problem \eqref{alg2min2} is hard to solve due to the complicated constraints  (\ref{alg2min2}a)-(\ref{alg2min2}c).
	In order to handle constraint (\ref{alg2min2}c), we use the expect value of  $c_{ik}(p_{ik})$ to represent the value of $c_{ik}(p_{ik})$, i.e., we have 
	\begin{equation}\label{sys_eq2_2_1}
		c_{ik}(p_{ik})= e^{-\frac{m\sigma^2}{p_{ik}h_{ik}}}.
	\end{equation}
 
	Without loss of generality, we consider the expression of accuracy function  $f(x,y)=\left(\frac{x}{y}\right)^{\alpha}$ with $\alpha$ being the constant to be determined in the simulations in the following analysis.
	Further substituting the expressions of $t_{nk}=\frac{d_{nk}}{r_{nk}}$ and $r_{nk}=B\log_2\left(1+\frac{p_{nk}h_{nk}}{\sigma^2}\right)$ as well as introducing slack variables $t_{nk}$ and $z_n$, problem \eqref{alg2min2} can be equivalent to
	\begin{subequations}\label{alg2min2_1}
		\begin{align}
			\mathop{\min }_{\boldsymbol p,\boldsymbol t,\boldsymbol z}\quad
			&\sum_{n=1}^N z_n \tag{\theequation}\\
			\textrm{s.t.}\quad\:\:
			&  z_n \geq x_{nk}t_{nk}, \quad \forall n\in \mathcal N,k\in\mathcal K,\\
			& t_{nk}\geq \frac{d_{nk}}{B\log_2\left(1+\frac{p_{nk}h_{nk}}{\sigma^2}\right)}, \quad \forall n\in \mathcal N,k\in\mathcal K,\\
			&\sum_{i=1}^n \sum_{k=1}^K d_{ik} e^{-\frac{m\sigma^2}{p_{ik}h_{ik}}} \geq A_n^{1/\alpha} \sum_{i=1}^n \sum_{k=1}^KD_{nk}, \quad \forall n\in\mathcal N,\\
			& \sum_{n=1}^N t_{nk} p_{nk} \leq Q_k, \quad \forall k\in\mathcal K, \\
			&\sum_{i=1}^n d_{ik}  e^{-\frac{m\sigma^2}{p_{ik}h_{ik}}}\leq  \sum_{i=1}^n D_{ik}, \quad \forall n\in \mathcal N,k\in\mathcal K,\\
			&0\leq p_{nk}\leq P_k, \quad \forall n\in\mathcal N, \\
   &0\leq t_{nk}\leq T_0, \quad \forall n\in\mathcal N,
		\end{align}
	\end{subequations}
	where $\boldsymbol z=[z_1,\cdots,z_N]^T$.
 	Problem \eqref{alg2min2_1} is  nonconvex due to  constraints 
	(\ref{alg2min2_1}b)-(\ref{alg2min2_1}e).
	To handle the nonconvexity of (\ref{alg2min2_1}b) and (\ref{alg2min2_1}d), we use variable  substitution. 
	Introducing a new variable $q_{nk}=t_{nk}p_{nk}$ and replacing $p_{nk}$ with $q_{nk}$, problem \eqref{alg2min2_1} becomes
	\begin{subequations}\label{alg2min2_2}
		\begin{align}
			\mathop{\min }_{\boldsymbol q,\boldsymbol t,\boldsymbol z}\quad
			&\sum_{n=1}^N z_n \tag{\theequation}\\
			\textrm{s.t.}\quad\:\:
			&  (\ref{alg2min2_1}\text{a}),~(\ref{alg2min2_1}\text{f}),~(\ref{alg2min2_1}\text{g}) \nonumber\\
			& t_{nk}B\log_2\left(1+\frac{q_{nk}h_{nk}}{\sigma^2t_{nk}}\right)\geq {d_{nk}}, \quad \forall n\in \mathcal N,k\in\mathcal K,\\
			&\sum_{i=1}^n \sum_{k=1}^K d_{ik} e^{-\frac{m\sigma^2t_{ik}}{q_{ik}h_{ik}}} \geq A_n^{1/\alpha} \sum_{i=1}^n \sum_{k=1}^KD_{ik}, \quad \forall n\in\mathcal N,\\
			& \sum_{n=1}^N q_{nk} \leq Q_k, \quad \forall k\in\mathcal K, \\
			&\sum_{i=1}^n d_{ik}  e^{-\frac{m\sigma^2t_{ik}}{q_{ik}h_{ik}}}\leq  \sum_{i=1}^n D_{ik}, \quad \forall n\in \mathcal N,k\in\mathcal K,
			%&0\leq p_{nk}\leq P_k,   \quad \forall n\in\mathcal N,\\
   			%& 0\leq t_{nk}\leq T_0, \quad \forall n\in\mathcal N,
      \end{align}
	\end{subequations}
	where both constraints (\ref{alg2min2_2}a) and (\ref{alg2min2_2}c) are convex now. As a result, it remains to solve the complex exponential expressions in constraints  (\ref{alg2min2_2}b) and (\ref{alg2min2_2}d).
 
	In order to handle the non-convexity of  constraints (\ref{alg2min2_2}b) and (\ref{alg2min2_2}d), we use the first-order Taylor series to approximate $e^{-\frac{m\sigma^2t_{ik}}{q_{ik}h_{ik}}}$, which can be given by 
	\begin{align}\label{alg2min2_2eq1}
		&\sum_{i=1}^n \sum_{k=1}^K d_{ik}
		e^{-\frac{m\sigma^2t_{ik}^{(m)}}{q_{ik}^{(m)}h_{ik}}} \Big( 1-\frac{m\sigma^2 }{q_{ik}^{(m)}h_{ik}}(t_{ik}-t_{ik}^{(m)})
		\nonumber\\&	
  +\frac{m\sigma^2t_{ik}^{(m)}}{(q_{ik}^{(m)})^2h_{ik}}
		(q_{ik}-q_{ik}^{(m)})
		\Big)
		\geq A_n^{1/\alpha} \sum_{i=1}^n \sum_{k=1}^KD_{ik}
	\end{align}
	and 
	\begin{align}\label{alg2min2_2eq2}
		&\sum_{i=1}^n d_{ik}  
		e^{-\frac{m\sigma^2t_{ik}^{(m)}}{q_{ik}^{(m)}h_{ik}}} \Big( 1-\frac{m\sigma^2 }{q_{ik}^{(m)}h_{ik}}(t_{ik}-t_{ik}^{(m)})
		\nonumber\\&	
  +\frac{m\sigma^2t_{ik}^{(m)}}{(q_{ik}^{(m)})^2h_{ik}}
		(q_{ik}-q_{ik}^{(m)})
		\Big)
		\leq  \sum_{i=1}^n D_{ik}, 
	\end{align}
	where $q_{ik}^{(m)}$ and $t_{ik}^{(m)}$ are respectively the values of  $q_{ik}$ and $t_{ik}$ in the $m$-th iteration. 
According to \eqref{alg2min2_2eq1} and \eqref{alg2min2_2eq2}, the convex linear inequality constraints are obtained. 
 
	Through replacing constraints (\ref{alg2min2_2}b) and (\ref{alg2min2_2}d) with \eqref{alg2min2_2eq1} and \eqref{alg2min2_2eq2} respectively,  problem \eqref{alg2min2_2} becomes the following convex problem:
 	\begin{subequations}\label{alg2min2_2_2}
		\begin{align}
			\mathop{\min }_{\boldsymbol q,\boldsymbol t,\boldsymbol z}\quad
			&\sum_{n=1}^N z_n \tag{\theequation}\\
			\textrm{s.t.}\quad\:\:
			&  (\ref{alg2min2_1}\text{a}),~(\ref{alg2min2_1}\text{f}),~(\ref{alg2min2_1}\text{g}) \nonumber\\
			& t_{nk}B\log_2\left(1+\frac{q_{nk}h_{nk}}{\sigma^2t_{nk}}\right)\geq {d_{nk}}, \quad \forall n\in \mathcal N,k\in\mathcal K,\\
			&\sum_{i=1}^n \sum_{k=1}^K d_{ik}
		e^{-\frac{m\sigma^2t_{ik}^{(m)}}{q_{ik}^{(m)}h_{ik}}} \Big( 1-\frac{m\sigma^2 }{q_{ik}^{(m)}h_{ik}}(t_{ik}-t_{ik}^{(m)})
		 \nonumber\\&	
 +\frac{m\sigma^2t_{ik}^{(m)}}{(q_{ik}^{(m)})^2h_{ik}}
		(q_{ik}-q_{ik}^{(m)})
		\Big)
	\nonumber\\&		\geq A_n^{1/\alpha} \sum_{i=1}^n \sum_{k=1}^KD_{ik}, \quad \forall n\in\mathcal N,\\
			& \sum_{n=1}^N q_{nk} \leq Q_k, \quad \forall k\in\mathcal K, \\
			&\sum_{i=1}^n d_{ik}  
		e^{-\frac{m\sigma^2t_{ik}^{(m)}}{q_{ik}^{(m)}h_{ik}}} \Big( 1-\frac{m\sigma^2 }{q_{ik}^{(m)}h_{ik}}(t_{ik}-t_{ik}^{(m)})
		 \nonumber\\&	
+\frac{m\sigma^2t_{ik}^{(m)}}{(q_{ik}^{(m)})^2h_{ik}}
		(q_{ik}-q_{ik}^{(m)})
		\Big)
		\nonumber\\&	\leq  \sum_{i=1}^n D_{ik},,k\in\mathcal K,
		%	&0\leq p_{nk}\leq P_k,   \quad \forall n\in\mathcal N,\\
   	%		& 0\leq t_{nk}\leq T_0, \quad \forall n\in\mathcal N,
		\end{align}
	\end{subequations}
 which can be effectively solved through using the existing convex optimization algorithms. 
 Trough solving problem \ref{alg2min2_2_2}, the solution of the original problem \eqref{alg2min2} can be effectively obtained. 
	
       \subsection{Data Offloading Subproblem}
	With given device scheduling and power control variables in problem \eqref{sys1min1},  the  data offloading subproblem can be rewritten as 
	\begin{subequations}\label{alg2min3}
		\begin{align}
			\mathop{\min }_{  \boldsymbol d }\quad
			&\sum_{n=1}^N  \max_{k\in\mathcal K}x_{nk}t_{nk} \tag{\theequation}\\
			\textrm{s.t.}\quad\:\:
			& a_n \geq A_n, \quad \forall n\in\mathcal N,\\
			& \sum_{n=1}^N t_{nk} p_{nk} \leq Q_k, \quad \forall k\in\mathcal K, \\
			&\sum_{i=1}^n d_{ik} c_{ik}(p_{ik})\leq  \sum_{i=1}^n D_{ik}, \quad \forall n\in \mathcal N,k\in\mathcal K,\\
			&d_{nk}\geq 0, \quad \forall n\in\mathcal N,   \\
			&0\leq t_{nk}\leq T_0, n\in\mathcal N.
		\end{align}
	\end{subequations}

	Substituting \eqref{sys_eq3} and \eqref{sys1accu} into problem  \eqref{alg2min3} yields
	\begin{subequations}\label{alg2min3_2}
		\begin{align}
			\mathop{\min }_{  \boldsymbol d }\quad
			&\sum_{n=1}^N  \max_{k\in\mathcal K}\frac{x_{nk}d_{nk} }{r_nk}\tag{\theequation}\\
			\textrm{s.t.}\quad\:\:
			&\sum_{i=1}^n \sum_{k=1}^K d_{ik} e^{-\frac{m\sigma^2}{p_{ik}h_{ik}}} \geq A_n^{1/\alpha} \sum_{i=1}^n \sum_{k=1}^KD_{ik}, \quad \forall n\in\mathcal N,\\
			& \sum_{n=1}^N \frac{p_{nk}d_{nk} }{r_nk} \leq Q_k, \quad \forall k\in\mathcal K, \\
			&\sum_{i=1}^n d_{ik} c_{ik}(p_{ik})\leq  \sum_{i=1}^n D_{ik}, \quad \forall n\in \mathcal N,k\in\mathcal K,\\
			&d_{nk}\geq 0, \quad \forall n\in\mathcal N,   \\
			&0\leq t_{nk}\leq T_0, \quad n\in\mathcal N,
		\end{align}
	\end{subequations}
	which is a linear programming problem and can be effectively solved by using the simplex method.

	\subsection{Algorithm Analysis}
	
	\begin{algorithm}[t]
		\caption{: Alternating Algorithm}
		\begin{algorithmic}[1]
			\STATE
			Initialize a feasible solution $( \boldsymbol x^{(0)}, \boldsymbol d^{(0)}, \boldsymbol p^{(0)}$ of problem (\ref{sys1min1})  and set $l=0$.
			\REPEAT
			\STATE With given $( \boldsymbol d^{(l)},  \boldsymbol p^{(l)})$, obtain the solution $\boldsymbol x^{(l+1)}$ of  problem \eqref{alg2min1}.
			\STATE With given $( \boldsymbol x^{(l+1)},  \boldsymbol d^{(l)})$, obtain the solution $\boldsymbol p^{(l+1)}$ of  problem \eqref{alg2min2}.
			\STATE With given $( \boldsymbol x^{(l+1)},  \boldsymbol p^{(l+1)})$, obtain the solution $\boldsymbol d^{(l+1)}$ of  problem \eqref{alg2min3}.
			\STATE Set $l=l+1$.
			\UNTIL {objective value (\ref{sys1min1}) converges}
		\end{algorithmic}
	\end{algorithm}
	
	Through alternatively solving subproblems \eqref{alg2min1}, \eqref{alg2min2}, and \eqref{alg2min3}, the overall procedure to obtain a solution of  problem (\ref{sys1min1}) can be shown  in Algorithm~1.
	Since the objective value (\ref{sys1min1})  is nonincreasing and the objective value (\ref{sys1min1})  has a limited lower bound (i.e., zero), Algorithm~1 always converges.

 The main complexity of Algorithm 1 lies in solving  subproblems \eqref{alg2min1}, \eqref{alg2min2}, and \eqref{alg2min3}. To solve subproblem \eqref{alg2min1}, the optimal solution can be obtained through Lemma 1 with complexity $\mathcal O(NK\tau)$ for given dual variables. Thus,  the complexity of solving subproblem \eqref{alg2min1} is  $\mathcal O(I_1NK\tau)$, where $I_1$ is the total number of iterations of using the dual method.  To solve subproblem \eqref{alg2min2}, the total complexity is $\mathcal O(I_1N^3 K^3)$ through solving a series of convex subproblems \eqref{alg2min2_2_2}, where $I_2$ is the total number of iterations using the successive convex approximation method.
 To solve subproblem \eqref{alg2min3}, the total complexity is  $\mathcal O(N^{2.5} K^{2.5})$ through the simplex method. 
As a result, the total complexity of Algorithm 1 is  $\mathcal O(I_0I_1NK\tau+I_0I_1N^3 K^3+I_0N^{2.5} K^{2.5})$, where $N_0$ denotes the total number of outer iterations.

	%Since problem (\ref{eem3max0}) is convex, the locally optimal solution obtained by Algorithm 1 is also the globally optimal solution to problem (\ref{eem3max0}).

\section{Optimization with One Physical Device}
In this section, we consider the case that the DT system only has one physical device. In the following of this paper, we omit the subscript $k$ without loss of generality since we only consider one physical device.   
For the case that $K=1$, problem \eqref{sys1min1}  can be formulated as 
 	\begin{subequations}\label{sys2min2}
		\begin{align}
			\mathop{\min }_{\boldsymbol x, \boldsymbol d, \boldsymbol p}\quad
			&\sum_{n=1}^N   x_{n}t_{n} \tag{\theequation}\\
			\textrm{s.t.}\quad\:\:
			& f\left(\sum_{i=1}^n  d_{i}c_{i}(p_{i}),\sum_{i=1}^n D_{i}\right) \geq A_n, \quad \forall n\in\mathcal N,\\
			& \sum_{n=1}^N t_{n} p_{n} \leq Q, \\
			&\sum_{i=1}^n d_{i} c_{i}(p_{i})\leq  \sum_{i=1}^n D_{i}, \quad \forall n\in \mathcal N,\\
			&\sum_{i=n}^{n+\tau } x_{i} \geq  \beta \tau, \quad \forall n\in\mathcal N,  \\
			&d_{n}\geq 0, \quad \forall n\in\mathcal N,   \\
              &x_{n}\in\{0,1\}, \quad \forall n\in\mathcal N,   \\
			&0\leq p_{n}\leq P,   \quad \forall n\in\mathcal N,\\
   			& 0\leq t_{n}\leq T, \quad \forall n\in\mathcal N,
      \end{align}
	\end{subequations}
where$\boldsymbol x=[x_{1},\cdots,x_{N}]^T$,
	$\boldsymbol d=[d_{1},\cdots,d_{N}]^T$,
 and
	$\boldsymbol p=[p_{1},\cdots,p_{N}]^T$.

It is generally hard to solve problem \eqref{sys2min2} due to the following two main difficulties.
The first difficulty is that problem \eqref{sys2min2} includes both integer and continuous variables, which introduces non-smoothness. 
The second difficulty lies in that the non-convexity constraints in (\ref{sys2min2}a)-(\ref{sys2min2}c). 
Due to the above two difficulties, the joint optimization design is challenging.

In the following, we solve problem \eqref{sys2min2} in an online manner. 
According to (\ref{sys2min2}d), the physical device needs to transmit information at least  $\lceil \beta \tau\rceil$ times in every $\tau+1$ time slots. 
For the sake of analysis, we consider the case that $\lceil \beta \tau\rceil=1$ and the  physical device needs to transmit at least once during every $\tau+1$ time slots. 
Based on this finding, we optimize the location of transmission time, i.e., $x_n=1$ in an inductive scheme. 
The proposed inductive scheme includes two phases: initial phase and recursion phase. 
\subsection{Initial Phase}
In the initial phase, the aim is to find the fist transmission time and the corresponding optimization problem can be formulated as 
 	\begin{subequations}\label{sys2min2_2}
		\begin{align}
			\mathop{\min }_{\{x_n,d_n, p_n\}}\quad
			&\sum_{n=1}^\tau   x_{n}t_{n} \tag{\theequation}\\
			\textrm{s.t.}\quad\:\:
			& f\left(\sum_{i=1}^{\tau+1}  d_{i}c_{i}(p_{i}),\sum_{i=1}^n D_{i}\right) \geq A_n, \nonumber\\&	\quad \forall n=1,\cdots,\tau+1,\\
			& \sum_{n=1}^{\tau+1} t_{n} p_{n} \leq \frac{(\tau+1)Q}{N}, \\
			&\sum_{i=1}^{\tau +1} x_{i} \geq 1, \quad \forall n=1,\cdots,\tau+1,  \\
			&d_{n}\geq 0, \quad \forall n=1,\cdots,\tau+1,\\
                & x_{n}\in\{0,1\}, \quad \forall n=1,\cdots,\tau+1,   \\
			&0\leq p_{n}\leq P,   \quad \forall n=1,\cdots,\tau+1,\\
   			& 0\leq t_{n}\leq T, \quad \forall n=1,\cdots,\tau+1,
      \end{align}
	\end{subequations}
where constraint (\ref{sys2min2}c) is omitted since the accuracy constraint (\ref{sys2min2}a) also reflects the minimum number of transmitted information requirement. 
To analyze the optimal solution of problem (\ref{sys2min2}), we have the following lemma.

\begin{lemma}
 For the optimal solution $(x_n^*,d_n^*, p_n^*)$ of problem (\ref{sys2min2}), we always have $\sum_{i=1}^{\tau +1} x_{i}^*= 1$. 
\end{lemma}

Lemma 2 can be proved by using the contradiction method. If the optimal solution of problem (\ref{sys2min2}) satisfying  $\sum_{i=1}^{\tau +1} x_{i}^*> 1$, we can always construct a new solution that $\sum_{i=1}^{\tau +1} x_{i}^*= 1$ and $x_{m}^*= 1$, $m=\arg_{i=1, \cdots, \tau+1} h_i$, 
with lower objective value, which contradicts that the solution is optimal. 

According to Lemma 2, one can always obtain the optimal transmission time slot, i.e., the time slot with the highest channel gain.
However, the original optimization problem \eqref{sys2min2} requires finding the transmission time for many $\tau+1$ time slots. 
Although the solution $x_{m}^*= 1$ is the optimal solution in the fist $\tau+1$ time slots, this does not guarantee that   $x_{m}^*= 1$ is optimal considering the whole $N$ time slots. 
As a result, we need to conduct the recursion phase to capture the relationship between multiple time slots. 
\subsection{Recursion Phase}
Assume that the previous transmission time is $m$, i.e.,  $x_{m}= 1$. 
In the recursion phase, our aim is to obtain the objective value if the next transmission time is $x_q=1$, $m+1\leq q \leq m+\tau+1$, i.e., the delay minimization problem can be formulated as 
 	\begin{subequations}\label{sys2min2_3}
		\begin{align}
			\mathop{\min }_{\{d_n, p_n\}}\quad
			&     t_{q} \tag{\theequation}\\
			\textrm{s.t.}\quad\:\:
			& f\left(  d_{q}c_{q}(p_{q}),\sum_{i=1}^n D_{i}\right) \geq A_n, \quad \forall n=1,\cdots,\tau+1,\\
			& t_{q} p_{q} \leq \frac{(\tau+1)Q}{N-m}, \\
			&d_{q}\geq 0, \quad \forall n=1,\cdots,\tau+1,\\
			&0\leq p_{q}\leq P,   \quad \forall n=1,\cdots,\tau+1,\\
   			& 0\leq t_{q}\leq T, \quad \forall n=1,\cdots,\tau+1.
      \end{align}
	\end{subequations}
According to \eqref{sys_eq1}, \eqref{sys_eq3} and (\ref{sys2min2_3}d), we have 
\begin{equation}\label{sys2min2_3eq1}
    t_q \geq \frac{d_q}{B\log_2\left(1+\frac{P h_{q}}{\sigma^2}\right)}.
\end{equation}
 Based on \eqref{sys_eq1}, \eqref{sys_eq3} and (\ref{sys2min2_3}d), we further have 
\begin{equation}\label{sys2min2_3eq2}
    t_q \geq \frac{Q}{\bar p_q},
\end{equation}
 where 
\begin{equation}\label{sys2min2_3eq3} 
   \frac{d_q \bar p_q}{B\log_2\left(1+\frac{\bar p_q h_{q}}{\sigma^2}\right)}=\frac{(\tau+1)Q}{N-m}.
\end{equation}
Solving \eqref{sys2min2_3eq3} with Lambert W function yields
\begin{equation}\label{sys2min2_3eq5} 
 \bar p_q=\frac{\sigma^2}{h_{q}}\left(
 \frac{Bh_qQ(\tau+1)}{-\ln2 d_q \sigma^2(N-m)}
 W\left(
 -\ln2
 \right)
 -1\right).
\end{equation}
Based on \eqref{sys2min2_3eq1} to \eqref{sys2min2_3eq5}
\begin{align}\label{sys2min2_3eq5}
    t_q^* =& \max\left\{ 
  \frac{d_q}{B\log_2\left(1+\frac{P h_{q}}{\sigma^2}\right)},\right.
  \\
  &\left.
  \frac{Q}{\frac{\sigma^2}{h_{q}}\left(
 \frac{Bh_qQ(\tau+1)}{-\ln2 d_q \sigma^2(N-m)}
 W\left(
 -\ln2
 \right)
 -1\right)}
    \right\}\triangleq T_{mq}.
\end{align}
 
\subsection{Overall Phase}
Based on \eqref{sys2min2_3eq5}, we show that the optimal delay is $T_{mq}$if the previous transmission time slot  is $m$ and the current transmission slot is $q$. 
Based on the expression of  $T_{mq}$, we can obtain 
\begin{equation}
 \sum_{n=1}^N   x_{n}t_{n}
 =\sum_{n=1}^N   x_{n} \sum_{m=n+1}^{\min\{N,n+\tau+1\}} x_m T_{nm}.
\end{equation}
As a result, 
the original optimization problem \eqref{sys2min2} can be reformulated as
 	\begin{subequations}\label{sys2min3}
		\begin{align}
			\mathop{\min }_{\boldsymbol x, \boldsymbol d }\quad
			&\sum_{n=1}^N    x_{n} \sum_{m=n+1}^{\min\{N,n+\tau+1\}} x_m T_{nm} \tag{\theequation}\\
			\textrm{s.t.}\quad\:\:
			& f\left(\sum_{i=1}^n  d_{i}c_{i}(p_{i}),\sum_{i=1}^n D_{i}\right) \geq A_n, \quad \forall n\in\mathcal N,\\
			&\sum_{i=n}^{n+\tau } x_{i} \geq  1, \quad \forall n\in\mathcal N,  \\
			&d_{n}\geq 0, \quad \forall n\in\mathcal N,   \\
              &x_{n}\in\{0,1\}, \quad \forall n\in\mathcal N.       \end{align}
	\end{subequations}
With given data offloading, the device scheduling belongs to the classic assignment problem, which can be effectively solved using the Hungarian algorithm.
With given user scheduling, the data offloading problem can be solved by using the successive convex approximation  algorithm. 

	\section{Simulation Results}
	In this section, we provide the simulation results of the proposed algorithm. 
	We consider a PN with $K$ physical devices uniformly distributed in a square area with size $200\textrm{m}\times200\textrm{m}$.
	The path loss model is $128.1+37.6\log_{10} d$ ($d$ is in km)
	and the standard deviation of shadow fading is $8$ dB \cite{yang2023energy}.
	The bandwidth is $B=1MHz$ and the noise power spectral density is  $N_0=-174$ dBm/Hz.
	Unless otherwise specified, we set the number of  physical devices $K=10$, 
	the total number of time slots $N=10$,
	equal accuracy requirement $A_1=\cdots=A_n=0.6$, equal arriving data  
	$D_{11}=\cdots=D_{NK}=300$ kbits,  parameter $\tau=3$, $K_0=5$,  equal constant $\beta_1=\cdots=\beta_K=1/3$, and equal maximum transmit power $P_1=\cdots=P_K=1$dBm.
	
	To show the effectiveness of the proposed algorithm, we consider the following two baselines: the random device selection algorithm (labeled as `Random'),  where the power control and data offloading are optimized by using the proposed algorithm, and the adaptive edge  association algorithm in \cite{lu2021adaptive} (labeled as `AEA'), where the device selection and data offloading are solved by using the proposed algorithm.

	\begin{figure}
		\centering
		\includegraphics[width=3.6in]{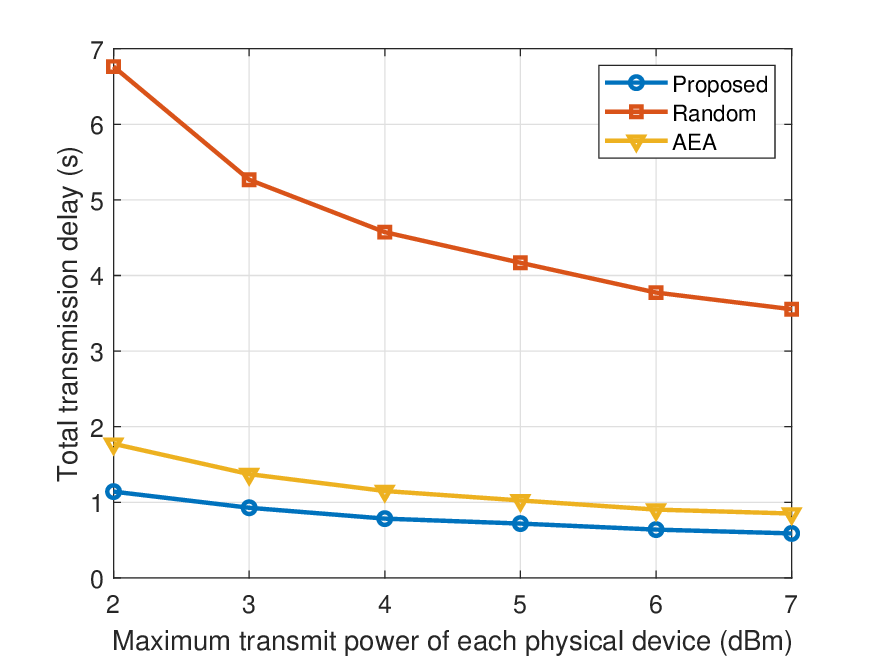}
		%\vspace{-1em}
		\caption{The transmission delay versus the maximum transmit power of  each physical device with $K=14$.}\label{fig12}
		%\vspace{-1em}
	\end{figure}
		
	\begin{figure}
		\centering
		\includegraphics[width=3.6in]{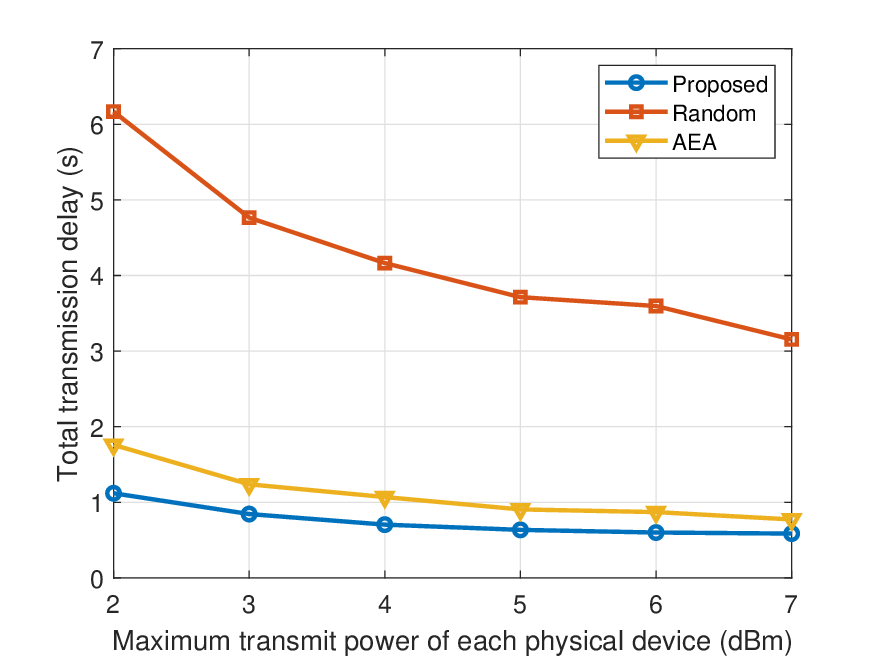}
 		\caption{The transmission delay versus the maximum transmit power of  each physical device with $K=10$.}\label{fig13}
 	\end{figure}

In Figs. \ref{fig12} and \ref{fig13}, we show the transmission delay versus the maximum transmit power of  each physical device with various number of physical devices. From both figures, it can be observed that the total transmission delay of all algorithms decreases with the maximum transmit power and the decreasing speed is decreasing. The reason is  that large maximum transmit power can allow the physical device to transmit with high power, thus reducing the transmission delay. Compared with AEA  in \cite{lu2021adaptive}, the proposed algorithm can yield lower transmission delay expecially when the maximum transmit power is small. The reason is that the proposed algorithm jointly optimizes multiple time slots, while the resource is not cooperatively  scheduled in AEA. 

	\begin{figure}
		\centering
		\includegraphics[width=3.6in]{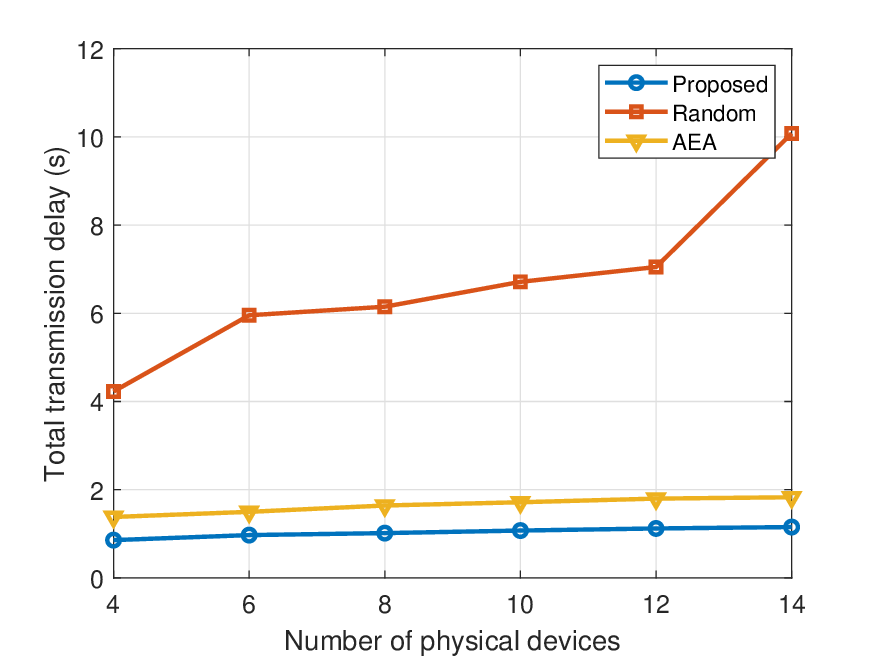}
 		\caption{The transmission delay versus the number of physical devices with maximum transmit power 2 dBm.}\label{fig10}
 	\end{figure}
		
	\begin{figure}
		\centering
		\includegraphics[width=3.6in]{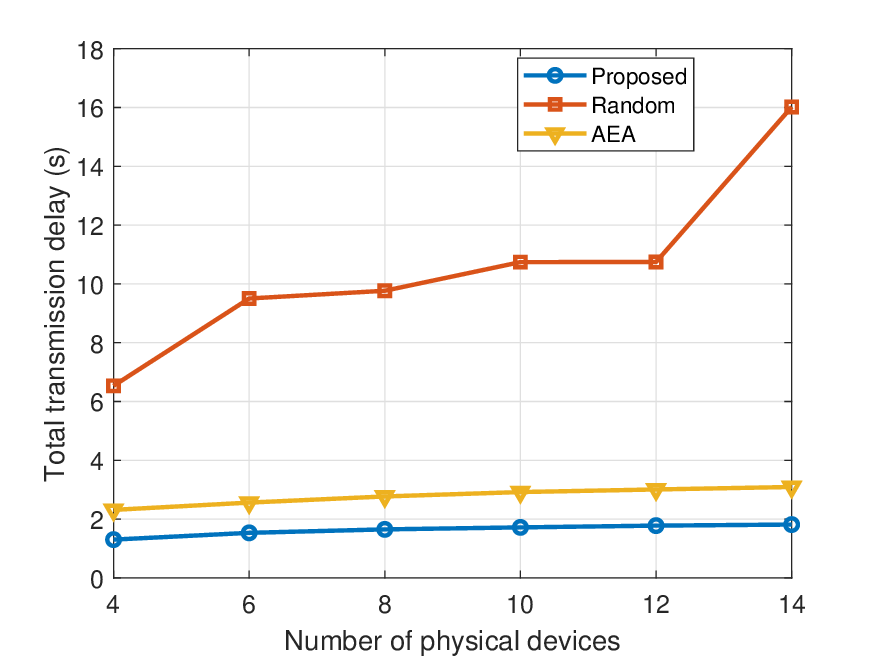}
 		\caption{The transmission delay versus the number of physical devices with maximum transmit power 1 dBm.}\label{fig11}
 	\end{figure}

The transmission delay versus the number of physical devices with different maximum transmit power is depicted in Figs. \ref{fig10} and \ref{fig11}. It can be found that the proposed algorithm always achieve the best performance. With the increase of number of physical devices, the total transmission delay of both the proposed and AEA algorithms slightly increases, while the total transmission delay of random algorithm dynamically increases. 
	
	\begin{figure}
		\centering
		\includegraphics[width=3.6in]{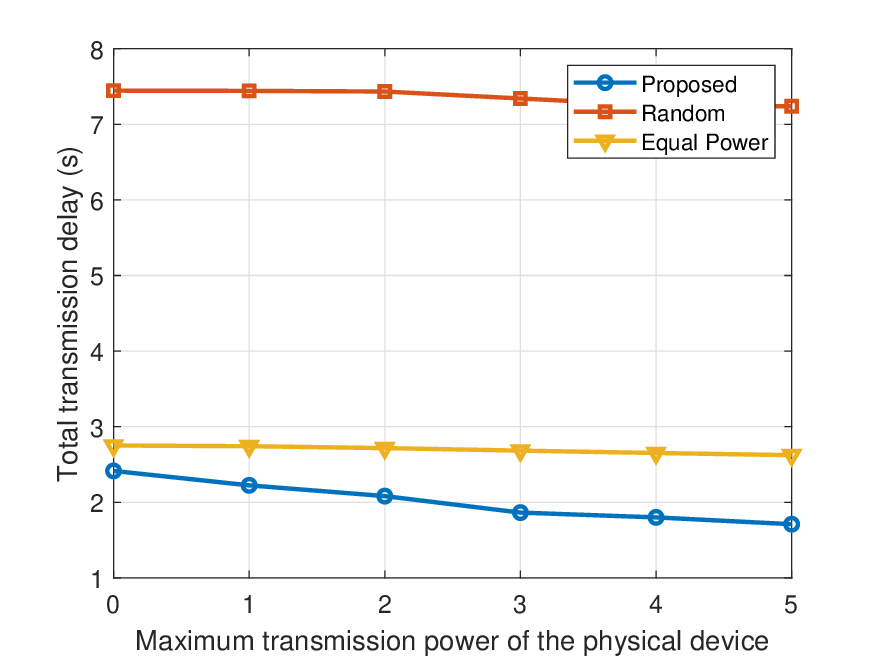}
 		\caption{The transmission delay versus the maximum transmit power of  each physical device.}\label{sim_fig1}
 	\end{figure}
	
	In the following Figs. \ref{sim_fig1} to \ref{sim_fig3}, we also compare the performance of the proposed algorithm with the equal power allocation scheme. 
	Fig. \ref{sim_fig1} shows the transmission delay versus the maximum transmit power of  each physical device. 
	According to this figure, it is observed that the proposed algorithm always achieve the best performance among all algorithms. 
	Compare to random device scheduling algorithm, the other two algorithms can great reduce the transmission delay, which indicates the superiority of device scheduling optimization.
	Compared to equal power allocation algorithm, the proposed algorithm can decrease the transmission delay by up to 51.2\% especially when the maximum transmission power is high. 
	The reason is that the proposed algorithm can dynamically allocate different power for each user based on the wireless channel gains to increase the overall transmission rate, thus leading to low transmission delay.  
	\begin{figure}
		\centering
		\includegraphics[width=3.6in]{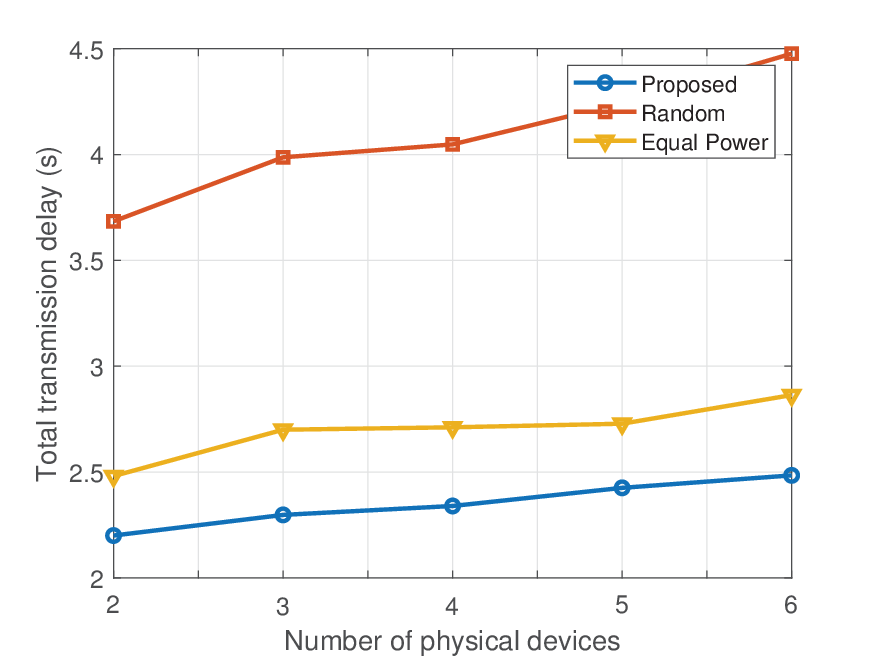}
		%\vspace{-1em}
		\caption{The transmission delay versus the number of physical devices.}\label{sim_fig2}
		%\vspace{-1em}
	\end{figure}
	
	The trend of transmission delay verses the number of physical devices is presented in Fig.~\ref{sim_fig2}.  
	It is found that all algorithm increases with the number of physical devices. This is because more physical devices means more data to offload, which can cause high transmission delay. 
	In this figure, we also can observe that both the proposed algorithm and equal power allocation algorithm can greatly reduce the transmission delay compared the random device scheduling algorithm. 
	
	\begin{figure}
		\centering
		\includegraphics[width=3.6in]{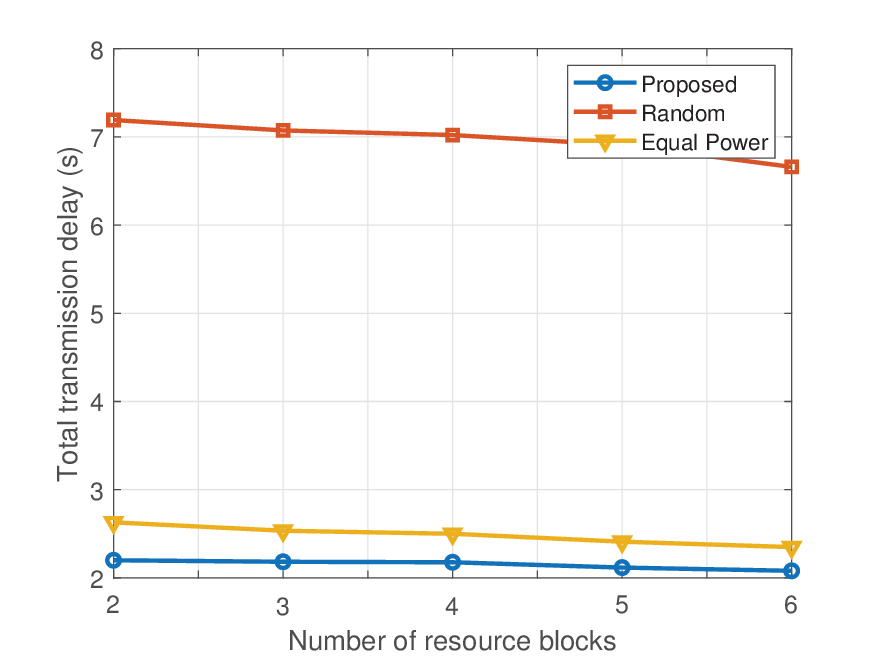}
		%\vspace{-1em}
		\caption{The transmission delay versus the number of resource blocks.}\label{sim_fig3}
		%\vspace{-1em}
	\end{figure}
	
	Fig. \ref{sim_fig2} illustrates the delay performance changes as the number of resource blocks. 
	From this figure, the transmission delay of all algorithms decreases with the number of resource blocks. 
	This is because more resource blocks ensure more devices to upload the data at each time slot, which can decrease the overall transmission time slots and result in low transmission delay.
	It is shown in Fig. 4 that the proposed algorithm is superior over the equal power allocation algorithm especially for small number of resource blocks.

	\section{Conclusions}
	In this paper, we have investigated the delay performance of DT over wireless networks. 
	We have formulated  a joint communication and computation problem so as to
	minimize the total  transmission delay of the network with considering both transmission energy and computation accuracy constraints.
	To solve this problem, we have proposed an alternating algorithm with solving three subproblems iteratively. 
	Numerical results have illustrated that the superiority of the proposed algorithm compared to the conventional schemes in terms of transmission delay, especially for large maximum   transmit power and small resource  blocks.

	\bibliographystyle{IEEEtran}
	\bibliography{MMM} 

\begin{thebibliography}{10}
\providecommand{\url}[1]{#1}
\csname url@samestyle\endcsname
\providecommand{\newblock}{\relax}
\providecommand{\bibinfo}[2]{#2}
\providecommand{\BIBentrySTDinterwordspacing}{\spaceskip=0pt\relax}
\providecommand{\BIBentryALTinterwordstretchfactor}{4}
\providecommand{\BIBentryALTinterwordspacing}{\spaceskip=\fontdimen2\font plus
\BIBentryALTinterwordstretchfactor\fontdimen3\font minus \fontdimen4\font\relax}
\providecommand{\BIBforeignlanguage}[2]{{%
\expandafter\ifx\csname l@#1\endcsname\relax
\typeout{** WARNING: IEEEtran.bst: No hyphenation pattern has been}%
\typeout{** loaded for the language `#1'. Using the pattern for}%
\typeout{** the default language instead.}%
\else
\language=\csname l@#1\endcsname
\fi
#2}}
\providecommand{\BIBdecl}{\relax}
\BIBdecl

\bibitem{khan2022digital}
L.~U. Khan, W.~Saad, D.~Niyato, Z.~Han, and C.~S. Hong, ``Digital-twin-enabled {6G}: Vision, architectural trends, and future directions,'' \emph{IEEE Communications Magazine}, vol.~60, no.~1, pp. 74--80, 2022.

\bibitem{jiang2022towards}
Y.~Jiang, L.~Peng, X.~Zhu, J.~Ma, X.~Yuan, Z.~Liu, and Q.~Wang, ``Towards metaverse with diversified applications: Service native architecture,'' \emph{Highlights in Science, Engineering and Technology}, vol.~24, pp. 31--39, 2022.

\bibitem{tang2022roadmap}
F.~Tang, X.~Chen, M.~Zhao \emph{et~al.}, ``The roadmap of communication and networking in {6G} for the metaverse,'' \emph{IEEE Wireless Commun.}, 2022.

\bibitem{cai2022compute}
Y.~Cai, J.~Llorca, A.~M. Tulino \emph{et~al.}, ``Compute-and data-intensive networks: The key to the metaverse,'' \emph{arXiv preprint arXiv:2204.02001}, 2022.

\bibitem{wang2022survey}
Y.~Wang, Z.~Su, N.~Zhang \emph{et~al.}, ``A survey on metaverse: Fundamentals, security, and privacy,'' \emph{IEEE Commun. Surveys \& Tutorials}, pp. 1--1, 2022.

\bibitem{xu2022full}
M.~Xu, W.~C. Ng, W.~Y.~B. Lim \emph{et~al.}, ``A full dive into realizing the edge-enabled metaverse: Visions, enabling technologies, and challenges,'' \emph{arXiv preprint arXiv:2203.05471}, 2022.

\bibitem{chang20226g}
L.~Chang, Z.~Zhang, P.~Li \emph{et~al.}, ``{6G}-enabled edge {AI} for metaverse: Challenges, methods, and future research directions,'' \emph{arXiv preprint arXiv:2204.06192}, 2022.

\bibitem{sun2016edgeiot}
X.~Sun and N.~Ansari, ``Edgeiot: Mobile edge computing for the internet of things,'' \emph{IEEE Communications Magazine}, vol.~54, no.~12, pp. 22--29, 2016.

\bibitem{saad2019vision}
W.~Saad, M.~Bennis, and M.~Chen, ``A vision of {6G} wireless systems: {A}pplications, trends, technologies, and open research problems,'' \emph{IEEE Net.}, vol.~34, no.~3, pp. 134--142, May/June 2020.

\bibitem{jones2020characterising}
D.~Jones, C.~Snider, A.~Nassehi, J.~Yon, and B.~Hicks, ``Characterising the digital twin: A systematic literature review,'' \emph{CIRP Journal of Manufacturing Science and Technology}, vol.~29, pp. 36--52, 2020.

\bibitem{liu2021review}
M.~Liu, S.~Fang, H.~Dong, and C.~Xu, ``Review of digital twin about concepts, technologies, and industrial applications,'' \emph{Journal of Manufacturing Systems}, vol.~58, pp. 346--361, 2021.

\bibitem{xu2022edge}
W.~Xu, Z.~Yang, D.~W.~K. Ng, M.~Levorato, Y.~C. Eldar \emph{et~al.}, ``Edge learning for {B5G} networks with distributed signal processing: Semantic communication, edge computing, and wireless sensing,'' \emph{arXiv preprint arXiv:2206.00422}, 2022.

\bibitem{zhang2018optimal}
Y.~Zhang, P.~You, and L.~Cai, ``Optimal charging scheduling by pricing for ev charging station with dual charging modes,'' \emph{IEEE Transactions on Intelligent Transportation Systems}, vol.~20, no.~9, pp. 3386--3396, 2018.

\bibitem{he2018preserving}
J.~He, L.~Cai, and X.~Guan, ``Preserving data-privacy with added noises: Optimal estimation and privacy analysis,'' \emph{IEEE Transactions on Information Theory}, vol.~64, no.~8, pp. 5677--5690, 2018.

\bibitem{han2013optimizing}
T.~Han and N.~Ansari, ``On optimizing green energy utilization for cellular networks with hybrid energy supplies,'' \emph{IEEE Transactions on Wireless Communications}, vol.~12, no.~8, pp. 3872--3882, 2013.

\bibitem{zhou2020pirate}
S.~Zhou, H.~Huang, W.~Chen, P.~Zhou, Z.~Zheng, and S.~Guo, ``Pirate: A blockchain-based secure framework of distributed machine learning in 5g networks,'' \emph{IEEE Network}, vol.~34, no.~6, pp. 84--91, 2020.

\bibitem{qu2021service}
Y.~Qu, H.~Dai, H.~Wang, C.~Dong, F.~Wu, S.~Guo, and Q.~Wu, ``Service provisioning for uav-enabled mobile edge computing,'' \emph{IEEE Journal on Selected Areas in Communications}, vol.~39, no.~11, pp. 3287--3305, 2021.

\bibitem{arnold2020cooperative}
E.~Arnold, M.~Dianati, R.~de~Temple, and S.~Fallah, ``Cooperative perception for 3d object detection in driving scenarios using infrastructure sensors,'' \emph{IEEE Transactions on Intelligent Transportation Systems}, vol.~23, no.~3, pp. 1852--1864, 2020.

\bibitem{le2022survey}
L.~Le~Mero, D.~Yi, M.~Dianati, and A.~Mouzakitis, ``A survey on imitation learning techniques for end-to-end autonomous vehicles,'' \emph{IEEE Transactions on Intelligent Transportation Systems}, 2022.

\bibitem{arnold2019survey}
E.~Arnold, O.~Y. Al-Jarrah, M.~Dianati, S.~Fallah, D.~Oxtoby, and A.~Mouzakitis, ``A survey on 3d object detection methods for autonomous driving applications,'' \emph{IEEE Transactions on Intelligent Transportation Systems}, vol.~20, no.~10, pp. 3782--3795, 2019.

\bibitem{zhan2020big}
Y.~Zhan, P.~Li, K.~Wang, S.~Guo, and Y.~Xia, ``Big data analytics by crowdlearning: Architecture and mechanism design,'' \emph{IEEE Network}, vol.~34, no.~3, pp. 143--147, 2020.

\bibitem{yang2023secure}
Z.~Yang, M.~Chen, G.~Li, Y.~Yang, and Z.~Zhang, ``Secure semantic communications: Fundamentals and challenges,'' \emph{arXiv preprint arXiv:2301.01421}, 2023.

\bibitem{yu20226g}
J.~Yu, A.~Alhilal, P.~Hui, and D.~H. Tsang, ``6{G} mobile-edge empowered metaverse: Requirements, technologies, challenges and research directions,'' \emph{arXiv preprint arXiv:2211.04854}, 2022.

\bibitem{akyildiz2022metaverse}
I.~F. Akyildiz, ``Metaverse: Challenges for extended reality and holographic-type communication in the next decade,'' in \emph{2022 ITU Kaleidoscope-Extended reality--How to boost quality of experience and interoperability}.\hskip 1em plus 0.5em minus 0.4em\relax IEEE, 2022, pp. 1--2.

\bibitem{chen2023big}
Z.~Chen, Z.~Zhang, and Z.~Yang, ``Big {AI} models for {6G} wireless networks: Opportunities, challenges, and research directions,'' 2023.

\bibitem{zhan2020learning}
Y.~Zhan, P.~Li, Z.~Qu, D.~Zeng, and S.~Guo, ``A learning-based incentive mechanism for federated learning,'' \emph{IEEE Internet of Things Journal}, vol.~7, no.~7, pp. 6360--6368, 2020.

\bibitem{wang2022task}
K.~Wang, J.~Jin, Y.~Yang, T.~Zhang, A.~Nallanathan, C.~Tellambura, and B.~Jabbari, ``Task offloading with multi-tier computing resources in next generation wireless networks,'' \emph{IEEE J. Sel. Areas Commun.}, vol.~41, no.~2, pp. 306--319, 2022.

\bibitem{wang2023task}
K.~Wang, D.~Niyato, W.~Chen, and A.~Nallanathan, ``Task-oriented delay-aware multi-tier computing in cell-free massive {MIMO} systems,'' \emph{IEEE J. Sel. Areas Commun.}, 2023.

\bibitem{kuruvatti2022empowering}
N.~P. Kuruvatti, M.~A. Habibi, S.~Partani, B.~Han, A.~Fellan, and H.~D. Schotten, ``Empowering {6G} communication systems with digital twin technology: {A} comprehensive survey,'' \emph{IEEE Access}, 2022.

\bibitem{lu2020low}
Y.~Lu, X.~Huang, K.~Zhang, S.~Maharjan, and Y.~Zhang, ``Low-latency federated learning and blockchain for edge association in digital twin empowered {6G} networks,'' \emph{IEEE Transactions on Industrial Informatics}, vol.~17, no.~7, pp. 5098--5107, 2020.

\bibitem{lu2021adaptive}
Y.~Lu, S.~Maharjan, and Y.~Zhang, ``Adaptive edge association for wireless digital twin networks in {6G},'' \emph{IEEE Internet of Things Journal}, vol.~8, no.~22, pp. 16\,219--16\,230, 2021.

\bibitem{pengnoo2020digital}
M.~Pengnoo, M.~T. Barros, L.~Wuttisittikulkij, B.~Butler, A.~Davy, and S.~Balasubramaniam, ``Digital twin for metasurface reflector management in {6G} terahertz communications,'' \emph{IEEE Access}, vol.~8, pp. 114\,580--114\,596, 2020.

\bibitem{zheng2021learning}
J.~Zheng, T.~H. Luan, L.~Gao, Y.~Zhang, and Y.~Wu, ``Learning based task offloading in digital twin empowered internet of vehicles,'' \emph{arXiv preprint arXiv:2201.09076}, 2021.

\bibitem{ruah2022digital}
C.~Ruah, O.~Simeone, and B.~Al-Hashimi, ``Digital twin-based multiple access optimization and monitoring via model-driven bayesian learning,'' \emph{arXiv preprint arXiv:2210.05582}, 2022.

\bibitem{hashash2022edge}
O.~Hashash, C.~Chaccour, and W.~Saad, ``Edge continual learning for dynamic digital twins over wireless networks,'' \emph{arXiv preprint arXiv:2204.04795}, 2022.

\bibitem{chen2021a}
M.~Chen, Z.~Yang, W.~Saad, C.~Yin, H.~V. Poor, and S.~Cui, ``A joint learning and communications framework for federated learning over wireless networks,'' \emph{IEEE Trans. Wireless Commun.}, vol.~20, no.~1, pp. 269--283, Jan. 2021.

\bibitem{boyd2004convex}
S.~Boyd, S.~P. Boyd, and L.~Vandenberghe, \emph{Convex optimization}.\hskip 1em plus 0.5em minus 0.4em\relax Cambridge university press, 2004.

\bibitem{yang2023energy}
Z.~Yang, M.~Chen, Z.~Zhang, and C.~Huang, ``Energy efficient semantic communication over wireless networks with rate splitting,'' \emph{IEEE Journal on Selected Areas in Communications}, vol.~41, no.~5, pp. 1484--1495, 2023.

\end{thebibliography}
	
\end{document}